\documentclass[runningheads]{llncs}

\usepackage{listings}
\usepackage{ltl}
\usepackage{multirow}
\usepackage{tabulary}
\usepackage{xcolor, colortbl}% http://ctan.org/pkg/xcolor
\usepackage{ltl}
\usepackage{tikz}
\usepackage[framemethod=tikz]{mdframed}
\usepackage{ifthen}
\usepackage{tikzscale}
\usepackage{pgfplots}
\tikzset{every picture/.style=semithick,
}%every axis plot/.append style=thick}

\tikzset{
    >=stealth,
    possible world/.style={circle,draw,thick,align=center},
    real world/.style={double,circle,draw,thick,align=center},
}
\usetikzlibrary{positioning,mindmap,automata,fit,backgrounds,shapes, 
arrows,shapes.misc,patterns}
\usetikzlibrary{decorations.pathmorphing}
\usetikzlibrary{decorations.markings}
\usetikzlibrary{calc}

\usepackage{ltl}
\usepackage{tabulary}
\usepackage{url}
\usepackage{todonotes}
\usepackage{url}
\usepackage{cite}
\usepackage{xspace}
\usepackage{graphicx}
\usepackage{subfigure}
\usepackage{caption}
\usepackage{color}
\usepackage{multicol}
\usepackage{hyphenat}
\usepackage{fancyvrb}
\usepackage{lipsum}
\usepackage{stmaryrd}
\usepackage{wrapfig}

\usepackage{hyperref}
\hypersetup{%
  colorlinks=true,           % use colored links
  allcolors=blue!70!black,   % use dark blue for all links
  pdfstartview=Fit,          % open PDF viewer with side-pane hidden.
  breaklinks=true,
  pdfauthor={B. Bonakdarpour and B. Finkbeiner},
  pdftitle={Program Repair for Hyperproperties}}

\usepackage[ruled,vlined,linesnumbered]{algorithm2e}

\DefineVerbatimEnvironment{code}{Verbatim}{fontsize=\small}

\begin{document}

\newcommand\modified[1]{#1}
\newcommand\bmodified[1]{#1}

\newcommand\frombaa[1]{\textcolor{red}{FROM BAA: #1}}
\newcommand\jyo[1]{\textcolor{blue}{Jyo: #1}}
\newcommand\chao[1]{\textcolor{blue}{Chao: #1}}
\newcommand\ufuk[1]{\textcolor{green}{Ufuk: #1}}
\newcommand\scott[1]{\textcolor{green}{Scott: #1}}
\newcommand\georgios[1]{\textcolor{brown}{Georgios: #1}}

\newcommand{\mypara}[1]{\vspace{0.5em} \noindent {\bf #1}.}
\newcommand{\myipara}[1]{\vspace{0.4em} \noindent {\em #1}.}

\newcommand{\lecps}{{\sc le}-{\sc cps}\xspace}
\newcommand{\lecpss}{{\sc le}-{\sc cps}\xspace}
\newcommand{\lecs}{{\sc lec}{\it s}\xspace}
\newcommand{\lec}{{\sc lec}\xspace}

\newcommand{\ignore}[1]{}

\newcommand\smin{\textcolor{red}{ins}}
\newcommand\smh{\textcolor{red}{hhs}}
\newcommand\smop{\textcolor{red}{ops}}
\newcommand\medin{\textcolor{red}{inm}}
\newcommand\medh{\textcolor{red}{hhm}}
\newcommand\medop{\textcolor{red}{opm}}
\newcommand\lgin{\textcolor{red}{inl}}
\newcommand\lgh{\textcolor{red}{hhl}}
\newcommand\lgop{\textcolor{red}{opl}}
\newcommand{\cur}{{\bf **}}
\newcommand{\current}{C}
\newcommand{\pending}{P}

\newcommand{\Phat}{\hat{P}}

\newcommand{\Paths}[2]{\mathit{Paths}^{#1}{(#2)}}
\newcommand{\fPaths}[2]{\mathit{Paths}^{#1}_{\mathit{fin}}{(#2)}}
\newcommand{\dom}{\mathit{dom}}
\newcommand{\dtmc}{\mathcal{M}}
\newcommand{\LTL}{\textsf{\small LTL}\xspace}
\newcommand{\CTL}{\textsf{\small CTL}\xspace}
\newcommand{\CTLstar}{\textsf{\small CTL$^*$}\xspace}
\newcommand{\PCTL}{\textsf{\small PCTL}\xspace}
\newcommand{\PCTLstar}{\textsf{\small PCTL$^*$}\xspace}
\newcommand{\HyperPCTL}{\textsf{\small HyperPCTL}\xspace}
\newcommand{\HyperLTL}{\textsf{\small HyperLTL}\xspace}
\newcommand{\HyperCTLstar}{\textsf{\small HyperCTL$^*$}\xspace}
\newcommand{\AFHyperLTL}{\mbox{AF-HyperLTL}\xspace}
\newcommand{\matching}{\mathcal{M}}
\newcommand{\topolgy}{\mathcal{T}}

\newcommand{\alphabet}{\mathrm{\Sigma}}
\newcommand{\states}{\mathrm{\Sigma}}
\newcommand{\statespace}{\states}
\newcommand{\Trace}{\mathsf{Traces}}
\newcommand{\trace}{t}
\newcommand{\qtrace}{\eta}
\newcommand{\sform}{\mathrm{\Phi}}
\newcommand{\pform}{\varphi}

\newcommand{\naturals}{\mathbb{N}_{>0}}
\newcommand{\naturalszero}{\mathbb{N}_{\geq 0}}

\newcommand{\AP}{\mathsf{AP}}

\newcommand{\Next}{\X}
\newcommand{\Finally}{\F}
\newcommand{\Globally}{\G}
\newcommand{\V}{\mathcal{V}}

\newcommand{\pr}{\mathbb{P}}
\renewcommand{\Pr}{\mathit{Pr}}

\newcommand{\emptyword}{\epsilon}

\newcommand{\init}{\mathit{init}}
\newcommand{\tpm}{\mathbf{P}}
\newcommand{\quant}{\mathbb{Q}}

\newcommand{\dbsim}{\mathit{dbSim}}
\newcommand{\res}{\mathit{res}}
\newcommand{\qout}{\mathit{qOut}}
\newcommand{\env}{\mathit{env}}
\newcommand{\fail}{\mathit{fail}}

\newcommand{\comp}[1]{\textsf{\small #1}}

\newcommand{\sigmakp}{$\mathsf{\Sigma^p_k}$\comp{-complete}\xspace}
\newcommand{\pikp}{$\mathsf{\Pi^p_k}$\comp{-complete}\xspace}

\newcommand\donotshow[1]{}

\newcommand{\ie}{i.e.\xspace}
\newcommand{\shield}{SHIELD\xspace}
\newcommand{\lestl}{\leccomp[{\sc stl}]}
\newcommand{\mulf}[1]{\multicolumn{2}{l}{#1}}

\newcommand{\z}{\cellcolor{black}}
\newcommand{\x}{\cellcolor{lightgray}}

\newcommand{\tru}{\mathsf{true}}
\newcommand{\false}{\mathsf{false}}
\newcommand{\inn}{\mathsf{in}}
\newcommand{\out}{\mathsf{out}}
\newcommand{\suffix}[2]{#1[#2,\infty]}
\newcommand{\F}{\LTLdiamond}
\newcommand{\G}{\LTLsquare}
\newcommand{\U}{\,\mathcal U\,}
\newcommand{\W}{\,\mathcal W\,}
\newcommand{\X}{\LTLcircle}

\newcommand{\States}{S}
\newcommand{\state}{s}
\newcommand{\trans}{\delta}
\newcommand{\kframe}{\mathcal{F}}
\newcommand{\krip}{\mathcal{K}}
\newcommand{\ktuple}{\langle S, s_\init, \trans, L \rangle}
\newcommand{\ktupleprime}{\langle S', s'_\init, \trans', L' \rangle}
\newcommand{\lang}{\mathcal{L}}

\newcommand{\pos}{\mathit{pos}}
\newcommand{\negt}{\mathit{neg}}
\newcommand{\PR}[2]{\mbox{\sf \small PR[#1, #2{}]}\xspace}

\newcommand{\GMNI}{\textsf{\small GMNI}\xspace}
\newcommand{\GNI}{\textsf{\small GNI}\xspace}

\pgfdeclarelayer{background}
\pgfdeclarelayer{foreground}
\pgfsetlayers{background,main,foreground}

\tikzset{
  invisible/.style={opacity=0, text opacity=0},
  visible on/.style={alt={#1{}{invisible}}},
  alt/.code args={<#1>#2#3}{%
    \alt<#1>{\pgfkeysalso{#2}}{\pgfkeysalso{#3}} % \pgfkeysalso doesn't change 
the path
  },
}

\newcommand\enrec[1]{%
  \tikz[baseline=(X.base)]
    \node (X) [draw, shape=circle, inner sep=0, fill=white] {\strut #1};}

\newcommand\encirclew[1]{%
  \tikz[baseline=(X.base)]
    \node (X) [draw, shape=circle, inner sep=0, fill=white] {\strut #1};}

\newcommand\encirclegr[1]{%
  \tikz[baseline=(X.base)]
    \node (X) [draw, shape=circle, inner sep=0, fill=gray] {\strut #1};}

 \newcommand\encirclep[1]{%
  \tikz[baseline=(X.base)]
    \node (X) [draw, shape=circle, inner sep=0, fill=pink] {\strut #1};}

\newcommand\encircle[1]{%
  \tikz[baseline=(X.base)]
    \node (X) [draw, shape=circle, inner sep=0, fill=green] {\strut #1};}

\newcommand\encircley[1]{%
  \tikz[baseline=(X.base)]
    \node (X) [draw, shape=circle, inner sep=0, fill=yellow] {\strut #1};}

\newcommand\encircler[1]{%
  \tikz[baseline=(X.base)]
    \node (X) [draw, shape=circle, inner sep=0, fill=red] {\strut #1};}

\tikzstyle{place}=[circle,thick,draw=blue!75,fill=blue!20,minimum size=6mm]
\tikzstyle{red place}=[place,draw=red!75,fill=red!20]
\tikzstyle{transition}=[rectangle,thick,draw=black, fill=black, 
minimum size=1mm]

\newcommand{\dec}{\mathtt{dec}}
\newcommand{\ses}{\mathtt{ses}}
\newcommand{\session}{\mathtt{session}}
\newcommand{\status}{\mathtt{status}}
\newcommand{\ntf}{\mathtt{ntf}}

\definecolor{mGreen}{rgb}{0,0.6,0}
\definecolor{mGray}{rgb}{0.5,0.5,0.5}
\definecolor{mPurple}{rgb}{0.58,0,0.82}
\definecolor{backgroundColour}{rgb}{0.95,0.95,0.92}

\lstdefinestyle{CStyle}{
    backgroundcolor=\color{backgroundColour},   
    commentstyle=\color{mGreen},
    keywordstyle=\color{magenta},
    numberstyle=\tiny\color{mGray},
    stringstyle=\color{mPurple},
    basicstyle=\footnotesize,
    breakatwhitespace=false,         
    breaklines=true,                 
    captionpos=b,                    
    keepspaces=true,                 
    numbers=left,                    
    numbersep=2pt,                  
    showspaces=false,                
    showstringspaces=false,
    showtabs=false,                  
    tabsize=2,
    language=C
}

% \SetKwInOut{Input}{Input}
% \SetKwInOut{Output}{Output}
% \SetKwProg{Fn}{Function}{}{}
% \DontPrintSemicolon

\newcommand{\eab}[1]{{\color{red}#1}}

\renewcommand{\topfraction}{0.96}
\renewcommand{\bottomfraction}{0.95}
\renewcommand{\textfraction}{0.1}
\renewcommand{\floatpagefraction}{1}
\renewcommand{\dbltopfraction}{.97}
\renewcommand{\dblfloatpagefraction}{.99}

\title{Program Repair for Hyperproperties}

\author{
  Borzoo Bonakdarpour\inst{1} %\orcidID{0000-0001-5608-8273}
  \and
  Bernd Finkbeiner\inst{2} %\orcidID{1111-2222-3333-4444}
}
%
%\authorrunning{S.~Stucki et al.}
\institute{% Department of Computing and Software,
Iowa State University, USA, \email{borzoo@iastate.edu} 
\and
Saarland University, Germany, \email{finkbeiner@cs.uni-saarland.de}
}

% \author{Borzoo Bonakdarpour\inst{1} \and Bernd Finkbeiner \and Shreya Agrawal 
% \and Sandeep Kulkarni}

%\institute{Iowa State University, USA \and University of Saarland, Germany}

\maketitle

\begin{abstract}
We study the repair problem for hyperproperties specified in the
temporal logic HyperLTL. Hyperproperties are system properties that
relate multiple computation traces. This class of properties includes
information flow policies like noninterference and observational
determinism.  The repair problem is to find, for a given Kripke
structure, a substructure that satisfies a given specification.
%
%Information-flow security properties often suffer from program refinement
%and repair, whereby a property holds for a system but is violated by
%a subset of its executions. This problem is dangerous, because
%it means that a resolution of the system's nondeterminism may render
%the system insecure, even if the original system is provably secure.
We show that the
repair problem is decidable for HyperLTL specifications and finite-state Kripke structures. We
provide a detailed complexity analysis for different fragments of
HyperLTL and different system types: tree-shaped, acyclic, and
general Kripke structures.

\end{abstract} 

\section{Introduction}
\label{sec:intro}

{\em Information-flow security} is concerned with the detection of unwanted 
flows of information from a set of variables deemed as secrets to 
another set of variables that are publicly observable. 
Information-flow security is foundational for some of 
the pillars of cybersecurity such as confidentiality, secrecy, and privacy.
Information-flow properties belong to the class of hyperproperties~\cite{cs10},
which generalize trace properties to sets of sets of traces.
Trace properties are usually insufficient, because information-flow properties
relate multiple executions. This also means that classic trace-based 
specification languages such as linear-time temporal logic (LTL) cannot be 
used directly to specify information-flow properties.
HyperLTL~\cite{cfkmrs14} is an extension of LTL with trace 
variables and quantifiers. HyperLTL can express information-flow properties
by simultaneously referring to multiple traces.
For example, \emph{noninterference}~\cite{gm82}
between a secret input $h$ and a public output $o$ can be specified in HyperLTL by stating that, for all pairs of traces $\pi$ and $\pi'$,  if the input is the same for all input variables $I$ except $h$, then the output
$o$ must be the same at all times:
\[
\forall\pi.\forall\pi'.~ \G \big(\!\!\!\bigwedge_{i\in I\setminus \{h\}}\! i_\pi = i_{\pi'}\big) ~\Rightarrow~ \G\, (o_\pi = o_{\pi'})
\]
Another prominent example is \emph{generalized noninterference} (GNI)~\cite{McCullough:1987:GNI},
which can be expressed as the following HyperLTL formula:
\[
\forall\pi.\forall\pi'.\exists\pi''.~\G\, (h_\pi = h_{\pi''}) ~\wedge~ \G\, (o_{\pi'} = o_{\pi''})
\]
The existential quantifier is needed to allow for nondeterminism. Generalized noninterference permits nondeterminism in the
low-observable behavior, but stipulates that low-security outputs may
not be altered by the injection of high-security inputs.

There has been a lot of recent progress in automatically
\emph{verifying}~\cite{frs15,FinkbeinerMSZ-CCS17,10.1007/978-3-319-96145-3_8,10.1007/978-3-030-25540-4_7}
and \emph{monitoring}~\cite{ab16,Finkbeiner2019,bsb17,bss18,fhst18,sssb19,10.1007/978-3-030-17465-1_7} HyperLTL specifications. The
automatic \emph{construction} of systems that satisfy a given set of
information-flow properties is still, however, in its infancy.  So
far, the only known approach is bounded synthesis~\cite{fhlst18,10.1007/978-3-030-25540-4_7},
which searches for an implementation up to a given bound on the number
of states.  While there has been some success in applying bounded
synthesis to systems like the dining cryptographers~\cite{journals/cacm/Chaum85}, this approach
does not yet scale to larger systems.  The general synthesis problem
(without the bound on the number of states) becomes undecidable as soon
as the HyperLTL formula contains two universal
quantifiers~\cite{fhlst18}. A less complex type of synthesis is {\em program 
repair}, where, given a model $\krip$ and a property $\varphi$, the 
goal is to construct a model $\krip'$, 
such that (1) any execution of $\krip'$ is also an execution of $\krip$, and 
(2) $\krip'$ satisfies $\varphi$. A useful application of program repair 
is {\em program sketching}, where the developer provides a program with 
``holes'' 
that are filled in by the synthesis algorithm~\cite{s13}. Filling a hole in a 
program sketch is a repair step that eliminates nondeterminism. While such a 
repair is guaranteed to preserve trace properties, it is well known that this 
is not the case in the context of information-flow security 
policies~\cite{j89}. In fact, this problem has not yet
been studied in the context of hyperproperties.

In this paper, we study the problem of automated program repair of 
finite-state systems with respect to HyperLTL specifications.
% Our 
% notion of refinement is the following: a system $S$ refines a system $S'$, if 
% the set of computations of $S$ is a subset of the set of computations of $S'$.
We provide a detailed analysis of the complexity of the repair problem
for different shapes of the structure:
we are interested in {\em general}, {\em acyclic}, and {\em tree-shaped} Kripke 
structures. The need for investigating the repair problem for tree-shaped 
and acyclic graphs stems from two reasons. First, many trace logs that can be 
used as a basis for example-based synthesis~\cite{at17} and repair are in 
the form of a simple linear collection of the traces seen so far. Or, for space 
efficiency, the traces are organized by common prefixes and assembled into a 
tree-shaped Kripke 
structure, or by common prefixes as well as suffixes assembled into an acyclic 
Kripke structure. The second reason is that tree-shaped and acyclic Kripke 
structures often occur as the natural representation of the state space of a 
protocol. For example, certain security protocols, such as authentication and 
session-based protocols (e.g., TLS, SSL, SIP) go through a finite sequence of 
\emph{phases}, resulting in an acyclic Kripke structure.

\begin{table*}[t]
\def\arraystretch{2.1}
\centering
\scalebox{0.85}{\newcolumntype{K}[1]{>{\centering\arraybackslash}p{#1}}
\begin{tabular}{|K{2cm}||K{3cm}||K{3cm}|K{.5cm}||K{4cm}|K{.5cm}|}
\hline

%& \multicolumn{3}{c||}{\cellcolor{black!15} \normalsize This paper} & \\
%\cline{2-4}
{\bf HyperLTL fragment} & {\bf Tree} & \multicolumn{2}{c||}{\bf Acyclic} & 
\multicolumn{2}{c|}{\bf General}\\

\hline\hline
$\mbox{E}^*$ & \multirow{4}{*}{\parbox[c]{2cm}{\centering \comp{L-complete}\\ 
{\em (Theorem~\ref{thrm:sys-tree-ea})}}} & 
\multicolumn{2}{c||}{\multirow{3}{*}{\parbox[c]{2cm}{\centering
\comp{NL-complete} \\ {\em (Theorems~\ref{thm:sys-acyc-e} and \ref{thm:sys-acyc-a})}}}} & \multicolumn{2}{c|}{\parbox[c]{2cm}{\centering \comp{NL-complete}\\ {\em (Theorem~\ref{thm:sys-general-e}) }}}\\
\cline{1-1}\cline{5-6}

$\mbox{A}^*$ &  & 
\multicolumn{2}{c||}{ } &  \multicolumn{2}{c|}{\parbox[c]{2cm}{\centering \comp{NP-complete}\\ {\em (Theorem~\ref{thm:sys-general-a})}}}\\
\cline{1-1}\cline{5-6}

$\mbox{EA}^*$ &  & 
\multicolumn{2}{c||}{ } & \multirow{2}{*}{\comp{PSPACE}} & \multirow{6}{*}{ \rotatebox[origin=c]{90}{\em 
(Theorem~\ref{thrm:system-general-EAk})} }\\
\cline{1-1}\cline{3-4}

$\mbox{E}^*\mbox{A}^*$ &  & 
\comp{$\mathsf{\Sigma}_{2}^p$} & \multirow{5}{*}{\rotatebox[origin=c]{90}{\em 
(Theorem~\ref{thrm:system-acyc-EAk1})} } & & \\
\cline{1-3}\cline{5-5}

$\mbox{AE}^*$ & {\parbox[c]{2cm}{\centering
\comp{P-complete} \\ {\em (Theorem~\ref{thrm:sys-tree-aeP})}}}
&  \multirow{2}{*}{\comp{$\mathsf{\Sigma}_{2}^p$-complete}} & 
 & \multirow{2}{*}{\comp{PSPACE-complete}} &  
\\
\cline{1-2}

$\mbox{A}^*\mbox{E}^*$ &  \multirow{4}{*}{{\parbox[c]{2cm}{\centering
\comp{NP-complete} \\ {\em (Corollary~\ref{cor:sys-tree-hltl})} }}} &
 &  
 & & \\
\cline{1-1}\cline{3-3}\cline{5-5}

$(\mbox{E}^*\mbox{A}^*)^k$ &   & 
\comp{$\mathsf{\Sigma}_{k}^p$-complete} & &
\multirow{2}{*}{\comp{$(k{-}1)$-EXPSPACE-complete}} & \\
\cline{1-1}\cline{3-3}

$(\mbox{A}^*\mbox{E}^*)^k$ &  & 
$\mathsf{\Sigma}_{k+1}^p$\comp{-complete} & & & \\
\cline{1-1}\cline{3-4}\cline{5-6}

$(\mbox{A}^*\mbox{E}^*)^*$ &  & 
\multicolumn{2}{c||}{\parbox[c]{3cm}{\centering \comp{PSPACE} \\ {\em 
(Corollary~\ref{cor:sys-acyclic-hltl}) }}} & 
\multicolumn{1}{c}{\parbox[c]{2cm}{\centering \comp{NONELEMENTARY}\\ {\em (Corollary~\ref{cor:sys-general-hltl})}}} & \\
\hline

\end{tabular}
}
\vspace{2mm}
\caption{Complexity of the HyperLTL repair problem in the size of 
the Kripke structure, where $k$ is the number of quantifier alternations in 
the formula.}
\label{tab:system}
%\vspace{-1cm}
\end{table*}

Table~\ref{tab:system} summarizes the contributions of this paper. It 
shows our results on the complexity of automated program repair with 
respect to different fragments of HyperLTL. The complexities are in the size 
of the Kripke structure. This \emph{system complexity} is the most 
relevant complexity in practice, because the system tends to be much larger 
than the specification. Our results show that the shape of the Kripke structure 
plays a crucial role in the complexity of the repair problem:

\begin{itemize}
\item {\bf Trees.} \ For trees, the complexity in the size of the Kripke 
structure does not go beyond \comp{NP}. The problem for the alternation-free 
fragment and the fragment with one quantifier alternation where the leading 
quantifier is existential is \comp{L-complete}. The problem for the fragment 
with one quantifier alternation where the leading quantifier is universal is 
\comp{P-complete}. The problem is \comp{NP-complete} for full HyperLTL. 

\item {\bf Acyclic graphs.} \ For acyclic Kripke structures, the complexity is 
\comp{NL-complete} for the alternation-free fragment and the fragment with one 
quantifier alternation where the leading quantifier is existential. The 
complexity is in the level of 
the polynomial hierarchy that corresponds to the number of quantifier 
alternations.

\item {\bf General graphs.} \ For general Kripke structures, the complexity is 
\comp{NL-complete} for the existential fragment and \comp{NP-complete} for the 
universal fragment. The complexity is \comp{PSPACE-complete} for the 
fragment with one quantifier alternation and $(k-1)$-\comp{EXPSPACE-complete} in 
the number $k$ of quantifier alternations.

\end{itemize}

We believe that the results of this paper provide the 
fundamental understanding 
of the repair problem for secure information flow and pave the way 
for further research on developing efficient and scalable techniques.

\paragraph*{Organization} The remainder of this paper is organized as follows.
In Section~\ref{sec:prelim}, we review Kripke structures and HyperLTL. We 
present a detailed motivating example in Section~\ref{sec:example}. The formal 
statement of our repair problem is in Section~\ref{sec:problem}.
Section~\ref{sec:tree} presents our results on the complexity of 
repair for HyperLTL in the size of tree-shaped Kripke structures. 
Sections~\ref{sec:acyclic} and~\ref{sec:general} present the 
results on the complexity of repair in acyclic and general graphs, 
respectively. We discuss related work in Section~\ref{sec:related}. We conclude 
with a discussion of future work in 
Section~\ref{sec:conclusion}. Detailed proofs are available in the full version of this paper.%in the 
%appendix.\todo{remove in the camready}

\section{Preliminaries}
\label{sec:prelim}

%We begin with a quick review of Kripke structures and HyperLTL.

\subsection{Kripke Structures}
\label{subsec:krip}

Let $\AP$ be a finite set of {\em atomic propositions} and $\alphabet = 
2^{\AP}$ be the {\em alphabet}. A {\em letter} is an element of $\alphabet$. A 
\emph{trace} $t\in \Sigma^\omega$ over alphabet $\alphabet$ is an infinite sequence of letters: $t = t(0)t(1)t(2) \ldots$

%For any $p \in \AP$, we say that $p$ holds in $s$ if and 
% only if $p \in s$. A {\em trace} is an infinite sequence of states. In the 
% context of {\em languages}, a trace over $\AP$ is essentially a {\em word} 
%over 
% the alphabet $\statespace = 2^{\AP}$, where a {\em letter} is a state. We 
% denote the set of all infinite words over $\statespace$ by 
% $\statespace^\omega$. A language $\lang$ over $\statespace$ is a subset of 
% $\statespace^\omega$.

% \noindent
% Kripke structures are defined as follows:
\begin{definition}
\label{def:kripke}
A {\em Kripke structure} is a tuple $\krip = \ktuple$,
where 

\begin{itemize}
 \item $\States$ is a finite set of {\em states};
 \item $\state_{\init} \in \States$ is the {\em initial state};
 \item $\trans \subseteq \States \times \States$ is a {\em transition 
relation}, and 
 \item $L: S \rightarrow \statespace$ is a {\em labeling function} on the 
states of 
$\krip$.
\end{itemize}
We require that 
for each $\state \in \States$, there exists $\state' \in \States$, such 
that $(\state, \state') \in \trans$.
% to  ensure  that  every  execution  of  a  
%Kripke  structure  can  always  be  extended to an infinite execution. \qed

\end{definition}

%In the graphical representation, we indicate the value of the labeling 
%function with Boolean expressions over the atomic propositions.
Figure~\ref{fig:kripke} shows an example Kripke structure where $L(s_{\init})= \{a\}, 
L(s_3)=\{b\}$, etc.
The \emph{size} of the Kripke structure is the number of 
its states. The directed graph $\kframe = \langle \States, \trans \rangle$ is 
called the {\em Kripke frame} of the Kripke structure $\krip$. A {\em loop} in 
$\kframe$ is a finite sequence $\state_0\state_1\cdots \state_n$, such that 
$(\state_i, \state_{i+1}) \in \trans$, for all $0 \leq i < n$, and $(\state_n, 
\state_0) \in \trans$. We call a Kripke frame {\em acyclic}, if the only loops 
are self-loops on otherwise terminal states, i.e., on states that have no other outgoing 
transition. See Fig.~\ref{fig:kripke} for an example. Since 
Definition~\ref{def:kripke} does not allow terminal states, we only consider 
acyclic Kripke structures with such added self-loops.

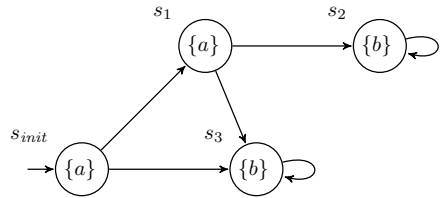
\begin{wrapfigure}{R}{6cm}
\centering
%\vspace{-1cm}
\scalebox{0.8}{
\begin{tikzpicture}[-,>=stealth',shorten >=.5pt,auto,node distance=2cm, 
semithick, initial text={}]

\node[initial, state] [text width=1em, text centered, minimum 
  height=2.5em](0) {\hspace*{-1.25mm}$\{a\}$};

\node [above left = 0.005 cm and 0.1 cm of 0](label){$s_{\init}$};

%\node [label={[label distance=0.25cm]-210:$s_{init}{:}$}] {};

\node[state, above right=of 0][text width=1em, text centered, minimum 
  height=2.5em] (1) {\hspace*{-1.25mm}$\{a\}$};

\node [above left = 0.005 cm and 0.1 cm of 1](label){$s_{1}$};

\node[state, right=of 1][text width=1em, text centered, minimum 
height=2.5em] (2) {\hspace*{-1.25mm}$\{b\}$};

\node [above left = 0.005 cm and 0.1 cm of 2](label){$s_{2}$};

\node[state, right=of 0][text width=1em, text centered, minimum 
height=2.5em] (3) {\hspace*{-1.25mm}$\{b\}$};

\node [above left = 0.005 cm and 0.1 cm of 3](label){$s_{3}$};

\draw[->]   
  (0) edge node (01 label) {} (1)
  (1) edge node (12 label) {} (3)
  (0) edge node (03 label) {} (3)
  (1) edge node (12 label) {} (2)
  (2) edge [loop right] node (22 label) {} (2)
  (3) edge [loop right] node (33 label) {} (3);
    
\end{tikzpicture}
}
\caption {An acyclic Kripke structure.}
\label{fig:kripke}
%\vspace{-1cm}
\end{wrapfigure}

We call a Kripke frame \emph{tree-shaped}, or, in short, a \emph{tree}, if 
every state $\state$ has a unique state $\state'$ with $(\state', \state) \in 
\trans$, except for the root node, which has no predecessor, and the leaf 
nodes, 
which, again because of Definition~\ref{def:kripke}, additionally have a 
self-loop but no other outgoing transitions.

A \emph{path} of a Kripke structure is an infinite sequence of states
$\state(0)\state(1)\cdots \in \States^\omega$, such that:

\begin{itemize}
 \item $\state(0) = \state_\init$, and
\item $(\state(i), \state({i+1})) \in \trans$, for all $i \geq 0$. 
\end{itemize}
A trace of a Kripke structure is a trace $t(0)t(1)t(2) \cdots \in 
\alphabet^\omega$, such that there exists a path $\state(0)\state(1)\cdots \in 
\States^\omega$ with $t(i) = L(\state(i))$ for all $i\geq 0$. We denote by 
$\Trace(\krip, \state)$ the set of all traces of $\krip$ with paths that start 
in state $\state \in \States$.

In some cases, the system at hand is given as a tree-shaped or acyclic Kripke 
structure. Examples include session-based security protocols and space-efficient 
execution logs, because
trees allow us to organize the traces according to common prefixes and acyclic 
graphs according to both common prefixes and common suffixes.

\subsection{The Temporal Logic HyperLTL}
\label{subsec:hltl}

HyperLTL~\cite{cfkmrs14} is an extension of linear-time temporal logic (LTL) for hyperproperties.
The syntax of HyperLTL formulas is defined inductively by the following grammar:
\begin{equation*}
\begin{aligned}
& \varphi ::= \exists \pi . \varphi \mid \forall \pi. \varphi \mid \phi \\
& \phi ::= \tru \mid a_\pi \mid \lnot \phi \mid \phi \vee \phi \mid \phi \
\U \, \phi \mid \X \phi
    \end{aligned}
\end{equation*}
where $a \in \AP$ is an atomic proposition and $\pi$ is a trace variable from 
an infinite supply of variables $\V$. The Boolean connectives $\neg$ and 
$\vee$ have the usual meaning, $\U$ is the temporal \emph{until} operator and 
$\X$ is the temporal \emph{next} operator. We also consider the usual derived 
Boolean connectives, such as $\wedge$, $\Rightarrow$, and $\Leftrightarrow$, 
and the derived temporal operators \emph{eventually} $\F\varphi\equiv 
\tru\,\U\varphi$ and \emph{globally} $\G\varphi\equiv\neg\F\neg\varphi$.
%and \emph{weak until} $\varphi \W \psi \equiv (\varphi \U \psi) \vee 
%\G\varphi$.
The quantified formulas $\exists \pi$ and $\forall \pi$ are read as `along some 
trace $\pi$' and `along all traces $\pi$', respectively.  

The semantics of HyperLTL is 
defined with respect to a trace assignment, a partial mapping~$\Pi \colon \V 
\rightarrow \alphabet^\omega$. The assignment with empty domain is denoted by 
$\Pi_\emptyset$. Given a trace assignment~$\Pi$, a trace variable~$\pi$, and 
a 
concrete trace~$t \in \alphabet^\omega$, we denote by $\Pi[\pi \rightarrow 
t]$ 
the assignment that coincides with $\Pi$ everywhere but at $\pi$, which is 
mapped to trace $t$. Furthermore, $\suffix{\Pi}{j}$ denotes the assignment 
mapping each trace~$\pi$ in $\Pi$'s domain to 
$\Pi(\pi)(j)\Pi(\pi)(j+1)\Pi(\pi)(j+2) \cdots $.
The satisfaction of a HyperLTL formula $\varphi$ over a trace assignment 
$\Pi$ and a set of traces $T \subseteq \alphabet^\omega$, denoted by $T,\Pi 
\models \varphi$, is defined as follows:
\[
\begin{array}{l@{\hspace{1em}}c@{\hspace{1em}}l}
  T, \Pi \models a_\pi & \text{iff} & a \in \Pi(\pi)(0),\\
  T, \Pi \models \neg \psi & \text{iff} & T, \Pi \not\models \psi,\\
  T, \Pi \models \psi_1 \vee \psi_2 & \text{iff} & T, \Pi \models \psi_1 
\text{ 
or } T, \Pi \models \psi_2,\\
  T, \Pi \models \X \psi & \mbox{iff} & T,\suffix{\Pi}{1} \models \psi,\\
  T, \Pi \models \psi_1 \U \psi_2 & \text{iff} &  \exists i \ge 0: 
T,\suffix{\Pi}{\modified{i}} \models \psi_2 \ \wedge\ 
\forall j \in [0, i): T,\suffix{\Pi}{j} \models \psi_1,\\
  T, \Pi \models \exists \pi.\ \modified{\psi} & \text{iff} & \exists t \in T: T,\Pi[\pi 
\rightarrow t] \models \psi,\\
  T, \Pi \models \forall \pi.\ \modified{\psi} & \text{iff} & \forall t \in T: T,\Pi[\pi 
\rightarrow t] \models \psi.
  \end{array}
\]
We say that a set $T$ of traces satisfies a sentence~$\varphi$, denoted by $T 
\models \phi$, if $T, \Pi_\emptyset \models \varphi$. If the set $T$ is 
generated by a Kripke structure $\krip$, we write $\krip \models \varphi$.

%Hyperproperties are a powerful formalism to express and reason about 
%information-flow security policies. For instance, consider {\em 
%observational determinism}~\cite{zm03}, which stipulates that in a concurrent 
%program if on every pair of computation traces where the observable inputs are 
%the same, also the observable outputs are the same. Observational determinism 
%can be expressed as the HyperLTL formula:
%\[
%\varphi_{\mathsf{obs}} = \forall \pi.\ \forall \pi'.\ \G\, (i_\pi 
%\Leftrightarrow i_{\pi'})\, 
%\Rightarrow\,  \G\, (o_\pi \Leftrightarrow o_{\pi'}),
%\]
%if two traces $\pi$ and $\pi'$ agree globally on input $i$, then they must also 
%globally agree on output $o$.

\section{Motivating Example}
\label{sec:example}

A real-life example that demonstrates the importance of the problem under 
investigation in this paper is the information leak in the EDAS Conference 
Management 
System\footnote{\url{http://www.edas.info}}, first reported in~\cite{ab16}.
The system manages the review process for papers submitted to conferences. Throughout this process, authors can check on the status of their papers, but should not learn whether or not the paper has been accepted until official notifications are sent out.
The system is correctly programmed to show status ``Pending'' before
notification time and ``Accept'' or ``Reject'' afterwards. The leak (which 
has since then been fixed) occurred through another status display, which 
indicates whether or not the paper has been scheduled for presentation in a 
session of the conference. Since only accepted papers get scheduled to sessions, 
this allowed the authors to infer the status of their paper.

The problem is shown in Table~\ref{tab:html1}. The first two rows 
show the output in the web interface for the authors regarding two papers submitted to a conference after their 
notification, where the first paper is accepted while the second is rejected. 
The last two rows show two other papers 
where the status is pending. The internal decisions on notification ($\ntf$), 
acceptance ($\dec$), and session ($\ses$), shown in the table with a gray 
background, are not part of the observable output and are added for the 
reader's convenience. However, by comparing the rows for the two pending papers, 
the authors can observe that the Session column values are not the same. Thus, 
they can still deduce that the first paper is rejected and the second paper is 
accepted.

\begin{table}
\begin{minipage}{.48\textwidth}
%\begin{table}[h]
\centering
\scalebox{0.9}{
\begin{tabular}{|c|c|c|c|c|c|}
\hline
   & \multicolumn{3}{c|}{\cellcolor{black!15} Internal Decisions} & \multicolumn{2}{c|}{\bf Output} \\
\cline{2-6}
   {\bf Paper}  & \cellcolor{black!15} $\ntf$ &
\cellcolor{black!15} $\dec$ & \cellcolor{black!15} 
$\ses$ & {\bf Status} & {\bf Session}\\
\hline
{\em foo1} & \cellcolor{black!15} $\tru$ & \cellcolor{black!15} $\tru$ & \cellcolor{black!15} $\tru$ & Accept & Yes \\
\hline
{\em bar1} & \cellcolor{black!15} \cellcolor{black!15} $\tru$ & \cellcolor{black!15} $\false$ & \cellcolor{black!15} $\false$ & Reject & No\\
\hline
{\em foo2} & \cellcolor{black!15} \cellcolor{black!15} $\false$ & \cellcolor{black!15} $\false$ & \cellcolor{black!15} $\false$ & Pending & No\\
\hline
{\em bar2} & \cellcolor{black!15} \cellcolor{black!15} $\false$ & \cellcolor{black!15} $\tru$ & \cellcolor{black!15} $\tru$ & Pending & Yes\\
\hline
\end{tabular}
}
%\end{table}
\caption{Output with leak.}
\label{tab:html1}
%\vspace{-1cm}
\end{minipage}
\begin{minipage}{.48\textwidth}
%\begin{table}[h]
\centering
\scalebox{0.9}{
\begin{tabular}{|c|c|c|c|c|c|}
  \hline
   & \multicolumn{3}{c|}{\cellcolor{black!15} Internal Decisions} & \multicolumn{2}{c|}{\bf Output} \\
\cline{2-6}
{\bf Paper} & \cellcolor{black!15} \cellcolor{black!15} $\ntf$ &
\cellcolor{black!15} $\dec$ & \cellcolor{black!15} 
$\ses$ & {\bf Status} & {\bf Session}\\
\hline
{\em foo1} & \cellcolor{black!15} $\tru$ & \cellcolor{black!15} $\tru$ & \cellcolor{black!15} $\tru$ & Accept & Yes \\
\hline
{\em bar1} & \cellcolor{black!15} \cellcolor{black!15} $\tru$ & \cellcolor{black!15} $\false$ & \cellcolor{black!15} $\false$ & Reject & No\\
\hline
{\em foo2} & \cellcolor{black!15} \cellcolor{black!15} $\false$ & \cellcolor{black!15} $\false$ & \cellcolor{black!15} $\false$ & Pending & No\\
\hline
{\em bar2} & \cellcolor{black!15} \cellcolor{black!15} $\false$ & \cellcolor{black!15} $\tru$ & \cellcolor{black!15} $\tru$ & Pending & No\\
\hline
\end{tabular}
}
\caption{Output without leak.}
\label{tab:html2}
%\vspace{-1cm}
\end{minipage}
% \caption{Sample outputs of the table generator. The Internal Decision column is 
% shown for the reader's convenience and is not part of the observable output.}
\end{table}

%
%\begin{wrapfigure}{R}{5.5cm}
\begin{figure}[t]
\centering
\vspace{-.2cm}
%\scalebox{0.8}{
\begin{lstlisting}[style=CStyle,tabsize=2,language=ML,basicstyle=\scriptsize,escapechar=/]
void Output(){
  bool ntf = GetNotificationStatus();
  bool dec = GetDecision();
  bool ses = getSession();
  
  string status =
    if (ntf)
    then if (dec)
         then "Accept"
         else "Reject"
    else "Pending";
                                
  string session =
    if (?) /\label{line:if}/                 
    then "Yes" 
    else "No"
  
  Print(status, session);
}
\end{lstlisting}
%}
\caption{Program sketch for a conference management system.}
\vspace{-.7cm}
\label{fig:code}
\end{figure}
%\end{wrapfigure}

The information leak in the EDAS system has previously been addressed by adding 
a monitor that detects such leaks~\cite{ab16,bf18}. Here, we instead eliminate 
the leak constructively. We use {\em program sketching}~\cite{s13} to 
automatically generate the code of our conference manager system. A program 
{\em sketch} expresses the high-level structure of an 
implementation, but leaves ``holes'' in place of the low-level details.
In our approach, the holes in a sketch are interpreted as nondeterministic choices. The repair eliminates nondeterministic choices in such a way that the specification becomes satisfied.

Figure~\ref{fig:code} shows a simple sketch for the EDAS example. The hole in 
the sketch (line~\ref{line:if}) is indicated by the question mark in the 
$\mathtt{if}$ statement. The replacement for the hole determines how the the 
value of the $\session$ output in the the web interface for the authors is 
computed. We wish to repair the sketch so that whenever two computations both 
result in $\status=\mbox{\tt "Pending"}$, the value of $\session$ is also the 
same. This requirement is expressed by the following HyperLTL formula:
\[
\begin{array}{r}
\varphi = \forall \pi. \forall \pi'. 
 \G \Big(\big((\status=\mbox{\tt "Pending"})_\pi \ \wedge\, 
(\status =\mbox{\tt "Pending"})_{\pi'}\big) \mbox{\qquad}\\ \Rightarrow\ 
 (\session_{\pi} \Leftrightarrow \session_{\pi'})\Big)
 \end{array}
 \]

%The secrets in this example are the internal decision ($\dec$) and the
%internal session assignment ($\ses$). The output of the web interface, i.e., the contents of variables $\status$ and $\session$ should remain independent of $\dec$ and $\ses$ until the notification flag $\ntf$ has been set. This requirement can be expressed as the following HyperLTL formula:
%\[
%\varphi = \forall \pi. \forall \pi'. 
% (\status_\pi \Leftrightarrow \status_{\pi'} \; \wedge \;
% \session_{\pi'} \Leftrightarrow \session_{\pi'}) \  \W\  \ntf
% \]
In this example, an \emph{incorrect} repair would be to replace the  hole in 
line~\ref{line:if} with $\ses$, which would result in the output of 
Table~\ref{tab:html1}. A \emph{correct} repair would be to replace the  hole 
with the Boolean condition $\ntf \, \wedge \, \ses$, which would result in the 
output of Table~\ref{tab:html2}. 

In the rest of the paper, we formally define the repair problem 
and study its complexity for different fragments of HyperLTL.

\section{Problem Statement}
\label{sec:problem}

The {\em repair problem} is the following decision problem. Let $\krip = 
\ktuple$ be a Kripke structure and $\varphi$ be a closed HyperLTL formula. Does 
there exist a Kripke structure $\krip' = \ktupleprime$ such that:

\begin{itemize}
 \item $\States' = \States$,
 \item $s'_\init = s_\init$,
 \item $\trans' \subseteq \trans$,
 \item $L' = L$, and
 \item $\krip' \models \varphi$?
\end{itemize}
In other words, the goal of the repair problem is to identify a Kripke 
structure $\krip'$, whose set of traces is a subset of the traces of $\krip$ that
satisfies $\varphi$. Note that since the witness to the decision problem 
is a Kripke structure, following Definition~\ref{def:kripke}, it is 
implicitly implied that in $\krip'$, for every state $s \in \States'$, there 
exists a state $s'$ such that $(s, s') \in \trans'$. In other words, \modified{the
repair} does not create a {\em deadlock} state. 

We use the following notation to distinguish the different 
variations of the problem:
\begin{center}
 \PR{\sf Fragment}{\mbox{\sf Frame Type}},
\end{center}
where

\begin{itemize}

 \item \comp{PR} is the {\em program repair} decision problem as described 
above;
 
\item \comp{Fragment} is one of the following for $\varphi$:

\begin{itemize}
% \item \comp{AF-HyperLTL} refers to the alternation-free fragment of HyperLTL 
%(i.e., $\exists^+\psi$ or $\forall^+\psi$);

 \item We use regular expressions to denote the order and pattern of repetition 
of quantifiers. For example, \comp{E$^*$A$^*$-HyperLTL} denotes the fragment, 
where an arbitrary (possibly zero) number of existential quantifiers is 
followed by an arbitrary (possibly zero) number of universal quantifiers. Also, 
\comp{$\mbox{A}\mbox{E}^+$-HyperLTL} means a lead universal quantifier followed 
by one or more existential quantifiers.
%
%\item We use trivial relations to bound the number of quantifiers. For example, 
\comp{E$^{\leq 1}$A$^*$-HyperLTL} denotes the fragment, where zero or one 
existential quantifier is followed by an arbitrary number of universal 
quantifiers.

\item \comp{(EA)$k$-HyperLTL}, for $k\geq 0$, denotes the fragment with $k$ 
alternations and a lead existential quantifier, where $k=0$ means an 
alternation-free formula with only existential quantifiers;

\item \comp{(AE)$k$-HyperLTL}, for $k\geq 0$, denotes the fragment with $k$ 
alternations and a lead universal quantifier, where $k=0$ means an 
alternation-free formula with only universal quantifiers,

\item \comp{HyperLTL} is the full logic HyperLTL, and

\end{itemize} 

\item \comp{Frame Type} is either \comp{tree}, \comp{acyclic}, or 
\comp{general}.

\end{itemize}

\section{Complexity of Repair for Tree-shaped Graphs}
\label{sec:tree}

In this section, we analyze the complexity of the program repair problem for trees.
This section is organized based on the rows in Table~\ref{tab:system}. We consider the following three HyperLTL fragments: (1) \comp{E$^*$A$^*$}, (2) \comp{$\mbox{A}\mbox{E}^*$}, and (3) the 
full logic.

\subsection{The \comp{E$^*$A$^*$} Fragment}

Our first result is that the repair problem for tree-shaped
Kripke structures can be solved in logarithmic time in the size of the Kripke 
structure for the fragment with only one quantifier alternation where the 
leading quantifier is existential. This fragment is the least expensive to deal 
with in tree-shaped Kripke structures and, interestingly, the 
complexity is the same as for the model checking problem~\cite{bf18}.

\begin{theorem}
  \label{thrm:sys-tree-ea}
\PR{E$^*$A$^*$-HyperLTL}{\mbox{tree}} is \comp{L-complete} in the size of the 
Kripke structure.
\end{theorem}

\begin{proof}
We note that the number of traces in a tree is bounded by the number of states, 
i.e., the size of the Kripke structure. The repair algorithm enumerates all 
possible assignments for the existential trace quantifiers, using, for each 
existential trace variable, a counter up to the number of traces, which 
requires only a logarithmic number of bits in size of the Kripke structure.
For each such assignment to the existential quantifiers, the algorithm steps 
through the assignments to the universal quantifiers, which again requires only 
a logarithmic number of bits in size of the Kripke structure. We consider only 
assignments with traces that have also been assigned to a existential 
quantifier.
For each assignment of the trace variables, we verify the formula, 
which can be done in logarithmic space~\cite{bf18}. If the verification is 
affirmative for all assignments to the universal variables, then
the repair consisting of the the traces assigned to the existential 
variables satisfies the formula.

In order to show completeness, we prove that the repair problem 
for the existential fragment is \comp{L-hard}.  
The \comp{L}-hardness for \PR{E$^*$-HyperLTL}{\mbox{tree}} and \linebreak 
\PR{A$^*$-HyperLTL}{\mbox{tree}} follows from the \comp{L-hardness} of 
ORD~\cite{et97}. ORD is the graph-reachability problem for directed 
line graphs. Graph reachability \linebreak from $s$ to $t$ can be checked with 
with the 
repair problems for $\exists \pi .\ \F (s_\pi \wedge \F t_\pi)$ or
$\forall \pi .\ \F (s_\pi \wedge \F t_\pi)$.\qed
\end{proof}

\subsection{The \comp{$\mbox{A}\mbox{E}^*$} Fragment}

We now consider formulas with one 
quantifier alternation where the leading quantifier is universal. The type of leading quantifier has a significant impact on the complexity of 
the repair problem: the complexity jumps from 
\comp{L-completeness} to \comp{P-completeness}, although the model checking 
complexity for this fragment remains \comp{L-complete}~\cite{bf18}.

\begin{theorem}
\label{thrm:sys-tree-aeP}

\PR{$\mbox{A}\mbox{E}^*$-HyperLTL}{\mbox{tree}} is \comp{P-complete} in the 
size of the Kripke structure.%\footnote{Detailed proofs appear in the 
%appendix.}\todo{remove from camready}

\end{theorem}

\noindent {\em Proof sketch.} \ 
Membership to \comp{P} can be shown by the following algorithm. For $\varphi 
= \forall \pi_1. \exists \pi_2 .\ \psi$, we begin by 
marking all the leaves. Then, in several rounds, we go through all marked 
leaves $v_1$ and instantiate $\pi_1$ with the trace leading to $v_1$. We then 
again go through all marked leaves $v_2$ and instantiate $\pi_2$ with the trace 
leading to $v_2$, and check $\psi$ on the pair of traces. If the check is 
successful for some instantiation of $\pi_2$, we leave $v_1$ marked, 
otherwise we remove the mark. When no more marks can be removed, we 
eliminate all branches of the tree that are not marked. For additional 
existential quantifiers, the number of rounds will increase linearly.

For the lower bound, we reduce the {\em Horn satisfiability} problem, which is 
\comp{P-hard}, to the repair problem for AE$^*$ formulas. We first 
transform the given Horn formula to one that every clause consists 
of two negative and one positive literals. We map this Horn formula to a 
tree-shaped Kripke structure and a constant-size HyperLTL formula. For example, 
formula $(\neg x_1 \vee \neg x_2 \vee f) \, \wedge \, (\neg x_3 \vee \neg f \vee 
x_4) \, \wedge \, (\neg x_2 \vee \neg 
x_2 \vee x_4) \, \wedge \, (\neg x_1 \vee \neg x_1   \vee \bot )$ is mapped to 
the Kripke structure in Fig.~\ref{fig:system-tree-ae}.

The Kripke structure includes one branch for each clause of 
the given Horn formula, where the length of each branch is in logarithmic order 
of the number of variables in the Horn formula. We use atomic propositions 
$\negt_1$ and $\negt_2$ to indicate negative literals and $\pos$ for the 
positive literal. We also include propositions $c$ and $h$ to mark 
each clause with a bitsequence. That is, for each clause $\{\neg x_{n_1} \vee 
\neg x_{n_2} \vee x_{p}\}$, we 
\begin{figure}[t]
\centering
\scalebox{.65}{
 \begin{tikzpicture}
 \tikzset{
    >=stealth,
    auto,
    % node distance=3.5cm,
%    font=\scriptsize,
    possible world/.style={circle,draw,thick,align=center},
    real world/.style={double,circle,draw,thick,align=center},
    minimum size=25pt
}

\coordinate (init) at (0, 0);

\node[draw,circle,text width=0.5cm,fill=white] (initstate) at ($ (init) + (3, 
1) $) {};

\draw[->] ($ (initstate) + (-.8,.5) $ ) -- (initstate);

% %-----------------------------------------------------------------

\foreach \i in {0,1,2,3}{
\foreach \j in {0,1,2,3}{
\draw [fill=white, align=center] ($ (init) + (-2.6, -.5) + ( \i*3.8,{(\j*-1.7) 
-1} ) $) ellipse (1.3cm and .45cm) node (v\i\j)
{\ifthenelse{\i=1 \AND \j=1}{$\{\negt_1\}$}{
\ifthenelse{\i=0 \AND \j=1}{$\{\negt_2\}$}{
\ifthenelse{\i=3 \AND \j=2}{}{
\ifthenelse{\i=1 \AND \j=2}{$\{\negt_2,\pos\}$}{
\ifthenelse{\(\i=0 \AND \j=2\) \OR \(\i=2 \AND \j=2\)}{$\{\pos\}$}{
\ifthenelse{\i=3 \AND \j=0}{$\{\negt_1,\negt_2\}$}{
\ifthenelse{\i=3 \AND \j=3}{}{
\ifthenelse{\i=2 \AND \j=1}{$\{\negt_1, \negt_2,h\}$}{}
\ifthenelse{\i=1 \AND \j=0}{$\{\negt_1, \negt_2, c\}$}{
\ifthenelse{\i=3 \AND \j=1}{$\{h\}$}{
\ifthenelse{\i=0 \AND \j=0}{$\{\negt_1,\pos\}$}{}
}
}}}}}}}
}};

}
\draw [->] (initstate) -- (v\i0.north);
}

% %-----------------------------------------------------------------

\foreach \i in {0,1,2,3}{
\foreach \j/\n in {0/1,1/2,2/3}{
\draw [->] (v\i\j) -- (v\i\n);

}}
 
% %-----------------------------------------------------------------

\foreach \i in {0,1,2,3}{
\foreach \j/\n in {0,1,2,3}{

\pgfmathsetmacro\ii{int(\i+1)};
\node (b) at (v\i\j)[xshift=-6mm,yshift=-7mm] 
{$b_{\ii_\j}$};

}}
 
% %-----------------------------------------------------------------

\foreach \i in {0,1,2,3} {
\path (v\i3) edge [loop below] (v\i3);
}

\end{tikzpicture}
 }
 \caption{The Kripke structure of the Horn formula.}
 \label{fig:system-tree-ae}
\end{figure}
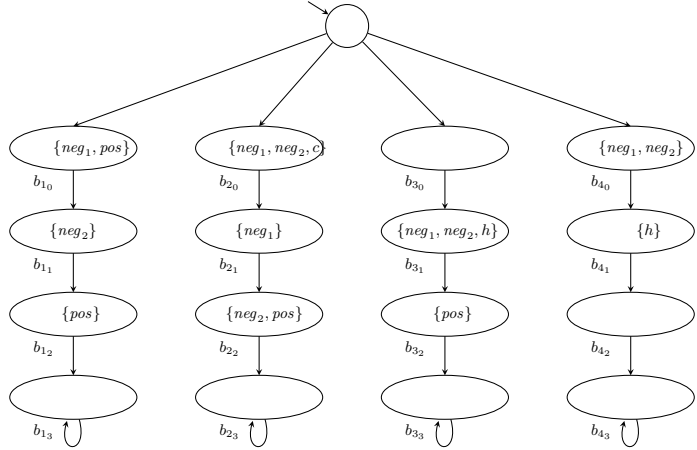
label states of its branch by atomic proposition $\negt_1$ according to the 
bitsequence of $x_{n_1}$, atomic proposition $\negt_2$ according to the 
bitsequence of $x_{n_2}$, and atomic proposition $\pos$ according to the 
bitsequence of $x_{p}$. We reserve values $0$ and $|X|-1$ for $\bot$ and 
$\bot$, respectively, where $X$ is the set of variables of the Horn formula. 
Finally, we use the atomic proposition $c$ to assign to each clause a number 
(represented as the bitsequence of valuations of $c$, starting with 
the lowest-valued bit; the position after the highest-level bit is marked by 
the occurrence of atomic proposition $h$, which does not appear anywhere else).

The HyperLTL formula enforces that (1) $\top$ is assigned to true, (2) $\bot$ 
is assigned to false, (3) all clauses are satisfied, and (4) if a positive 
literal $l$ appears on some clause in the repaired Kripke structure, then all 
clauses with $l$ must be preserved by the repair.\qed

\subsection{The Full Logic}

We now turn to full HyperLTL. We first show that the repair problem is in \comp{NP}. 

\begin{theorem}
\label{thrm:sys-tree-aae-upper}
\PR{HyperLTL}{\mbox{tree}} is in \comp{NP} in the size of the 
Kripke structure.
\end{theorem}

\begin{proof}
We nondeterministically 
guess a solution $\krip'$ to the repair problem. Since determining 
whether or not $\krip' \models \varphi$ can be solved in logarithmic 
space~\cite{bf18}, the repair problem is in \comp{NP}.\qed
\end{proof}

For the lower bound, the intuition is that an additional leading 
universal quantifier allows us to encode full Boolean satisfiability,
instead of just Horn satisfiability as in the previous section. Interestingly, 
the model checking problem remains \comp{L-complete} for this fragment~\cite{bf18}.

\begin{theorem}
\label{thrm:sys-tree-aae-lower}
\PR{AAE-HyperLTL}{\mbox{tree}} is \comp{NP-hard} in the size of the 
Kripke structure.
\end{theorem}

\noindent {\em Proof sketch.} \
We map an instance of the 3SAT problem to a Kripke structure and a HyperLTL 
formula. Figure~\ref{fig:sys-tree-aae} shows an example, where each clause in 
3SAT is mapped to a distinct branch and each literal in the clause 
is mapped to a distinct sub-branch. We label positive and negative 
literals by $\pos$ and $\negt$, respectively. Also, propositions $c$ and $h$ 
are used to mark the clauses with bitsequences in the same fashion as 
in the construction of proof of Theorem~\ref{thrm:sys-tree-aeP}. The HyperLTL 
formula $\varphi_{\mathsf{map}}$ ensures that (1) 
at least one literal in each clause is true, (2) a literal is not assigned to 
two values, and (3) all clauses are preserved during repair:
\[
\begin{array}{l}
\varphi_{\mathsf{map}} = \forall \pi_1.\forall \pi_2.\exists \pi_3.  \Bigg[\G 
\big(\neg \pos_{\pi_1} \, \vee \, \neg \negt_{\pi_2}\big)\Bigg] \; \wedge \\
{\X} \Bigg[\Big(\big(c_{\pi_2} {\wedge} \neg c_{{\pi_3}}\big) \, \U \, \big(\neg 
c_{\pi_2} \wedge c_{\pi_3} \wedge \X ((c_{\pi_2} \leftrightarrow c_{\pi_3})\, 
\mathcal 
U\, h_{{\pi_2}})\big)\Big) \, \vee\,  (c_{\pi_2} {\wedge} \neg c_{\pi_3}) \, \U \, 
h_{{\pi_2}}\Bigg]
\end{array}
\]
The answer to the 3SAT problem is affirmative 
if and only if a repair exists for the mapped Kripke structure with respect 
to formula $\varphi_{\mathsf{map}}$.\qed

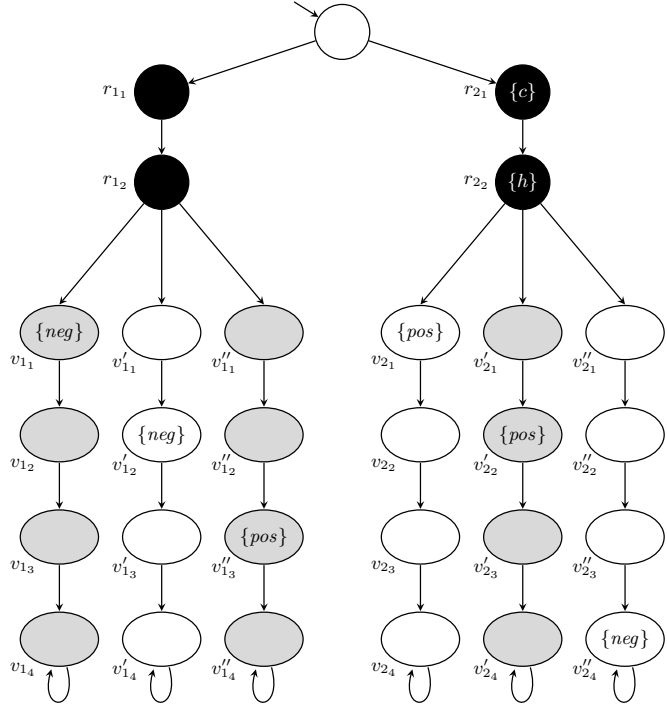
\begin{figure}[t]
\centering
\scalebox{.8}{
\begin{tikzpicture}
 
 \tikzset{
    >=stealth,
    %auto,
    %node distance=3.5cm,
%    font=\scriptsize,
    possible world/.style={circle,draw,thick,align=center},
    real world/.style={double,circle,draw,thick,align=center},
    minimum size=26pt
}

\coordinate (init) at (0, 0);

\node[draw,circle,text width=0.5cm,fill=white] (initstate) at ($ (init) + (3, 
1) $) {};

\draw[->] ($ (initstate) + (-.8,.5) $ ) -- (initstate);

%-----------------------------------------------------------------

\foreach \i in {0, 1} {
\foreach \j  in {0, 1} {
\node[draw,circle, minimum size=9mm ,fill=black] (r\i\j) at ($ (init) + 
(\i*6,\j*-1.5) $) [text=white]{\ifthenelse{\i=1 \AND 
\j=0}{$\{c\}$}{\ifthenelse{\i=1 \AND \j=1}{$\{h\}$}{}}};

\pgfmathsetmacro\ii{int(\i+1)};
\pgfmathsetmacro\jj{int(\j+1)};
\node (a) at (r\i\j.west)[xshift=-3mm] {$r_{\ii_\jj}$};
}}

%-----------------------------------------------------------------

\foreach \i in {0, 1} {
\foreach \j/\n in {0/1} {
\draw [->] (r\i\j) -- (r\i\n);
}}

%-----------------------------------------------------------------

\draw[->] (initstate) -- (r00);
\draw[->] (initstate) -- (r10);

%-----------------------------------------------------------------

\foreach \i in {0,1}{
\foreach \k in {0, 1, 2}{
\foreach \j in {1, 2, 3, 4}{
\ifthenelse{\(\i=0 \AND \k=1\) \OR \(\i=1 \AND \k=0\) \OR \(\i=1 \AND \k=2\)} 
{\def\clr{white}} {\def\clr{black!15}};
\draw [fill=\clr] ($ (init) + 
(r01)+ (0, -.8) + ( {(\i*6)+(\k*1.7-1.7)},\j*-1.7) $) ellipse (.65cm and .45cm) 
node 
(v\i\k\j)
{\ifthenelse{\(\i=0 \AND \j=3 \AND \k=2\) \OR \(\i=1 \AND \j=1 \AND \k=0\) \OR 
\(\i=1 \AND \j=2 \AND \k=1\)}{$\{\pos\}$}
{\ifthenelse{\(\i=0 \AND \j=1 \AND \k=0\) \OR \(\i=0 \AND \j=2 \AND \k=1\) \OR 
\(\i=1 \AND \j=4 \AND \k=2\) }{$\{\negt\}$}{}}{}};

\pgfmathsetmacro\ii{int(\i+1)};
\node (b) at (v\i\k\j)[xshift=-6mm,yshift=-5mm] 
{\ifthenelse{\k=0}{$v_{\ii_\j}$}{\ifthenelse{\k=1}{$v'_{\ii_\j}$}{$v''_{\ii_\j}$
} }};

}
\draw [->] (r\i1) -- (v\i\k1.north);
}

}

%-----------------------------------------------------------------

\foreach \i in {0,1}{
\foreach \k in {0, 1, 2}{
\foreach \j/\n in {1/2,2/3,3/4}
{
\draw [->] (v\i\k\j) -- (v\i\k\n);
}}}

%-----------------------------------------------------------------

\foreach \i in {0,1} {
\foreach \k in {0,1,2} {
\path (v\i\k4) edge [loop below] (v\i\k4);
}}

\end{tikzpicture}
}
\caption{The Kripke structure for the 3SAT formula $(\neg x_1 \vee \neg x_2 
\vee x_3) \ \wedge \ (x_1 \vee x_2 \vee \neg x_4)$. The truth assignment $x_1 = 
\tru$, $x_2 = \false$, $x_3 = \false$, $x_4 = \false$ renders the tree with 
white branches, i.e., the grey branches are removed during repair.}
\label{fig:sys-tree-aae}
\vspace{-.5cm}
\end{figure}

\begin{corollary}
  \label{cor:sys-tree-hltl}
The following are \comp{NP-complete} in the size of the 
Kripke structure: \PR{A$^*$E$^*$HyperLTL}{\mbox{tree}}, 
\PR{$\mbox{(EA)}^k\mbox{-HyperLTL}$}{\mbox{tree}}, 
\PR{$\mbox{(AE)}^k\mbox{-HyperLTL}$}{\mbox{tree}}, and
\PR{HyperLTL}{\mbox{tree}}.
\end{corollary}

\section{Complexity of Repair for Acyclic Graphs}
\label{sec:acyclic}

We now turn to acyclic graphs. Acyclic Kripke structures
are of practical interest, because
certain security protocols, in particular authentication algorithms,
often consist of sequences of phases with no
repetitions or loops.
We develop our results first for the alternation-free fragment, then for formulas with quantifier alternation.

\subsection{The Alternation-free Fragment}

We start with the existential fragment. The complexity of the repair problem 
for this fragment is interestingly the same as the model checking problem.

\begin{theorem}
\label{thm:sys-acyc-e}
\PR{$\mbox{E}^*$-HyperLTL}{\mbox{acyclic}} is \comp{NL}-complete in the size of 
the 
Kripke structure.
\end{theorem}

\begin{proof}
  For existential formulas, the repair problem is equivalent to the model 
checking problem. A given Kripke structure satisfies the formula iff it has a 
repair. If the formula is satisfied, the repair is simply the original 
Kripke structure.
  Since the model checking problem for existential formulas over acyclic graphs 
is  \comp{NL}-complete~\cite[Theorem 2]{bf18}, the same holds for the repair 
problem.\qed
\end{proof}

We now switch to the universal fragment.

\begin{theorem}
   \PR{A$^*$-HyperLTL}{\mbox{acyclic}} and
  \PR{EA$^*$-HyperLTL}{\mbox{acyclic}}
  is \linebreak \comp{NL}-complete in the size of the 
Kripke structure.
\label{thm:sys-acyc-a}
\end{theorem}

\begin{proof}
 To solve the repair problem of a HyperLTL formula \linebreak $\varphi = 
\exists \pi.\ 
\forall \pi_1. \forall \pi_2 \ldots \forall \pi_m.\ \psi(\pi, \pi_1, \pi_2, 
\ldots, \pi_m)$ with at most one existential quantifier, which appears as the 
first quantifier,
it suffices to find a single trace $\pi$ that satisfies $\psi(\pi, \pi, \ldots, 
\pi)$: suppose there exists a repair that satisfies $\varphi$ and that has 
more than one path, then any repair that only preserves one of these paths 
also satisfies the universal formula $\varphi$.
For the upper bound, we nondeterministically guess a path for the trace $\pi$  
and remove all other paths. Since the length of the path is bounded by the size 
of the acyclic Kripke structure, we can guess the path using logarithmically 
many bits for a counter measuring the length of the path.

\comp{NL}-hardness of \PR{A$^*$-HyperLTL}{\mbox{acyclic}} follows from the \comp{NL-hardness} of the 
graph-reachability 
problem for ordered graphs~\cite{LENGAUER199263}. Ordered graphs are acyclic 
graphs with a vertex numbering that is a topological sorting of the vertices. 
We express graph reachability from 
vertex $s$ to vertex $t$ as the repair problem of the universal formula 
$\forall \pi .\ \F (s_\pi \wedge \F t_\pi)$.\qed
\end{proof}

\subsection{Formulas with Quantifier Alternation}

Next, we consider formulas where the number of quantifier
alternations is bounded by a constant $k$. We show that changing
the frame structure from 
trees to acyclic graphs results in a
significant increase in complexity (see 
Table~\ref{tab:system}). The complexity of the repair problem is 
similar to the model checking problem, with the repair problem being one level 
higher in the polynomial hierarchy (cf.~\cite{bf18}).

\begin{theorem} For $k\geq 2$,
\label{thrm:system-acyc-EAk1}
\PR{\mbox{(EA)}$k$\mbox{-HyperLTL}}{\mbox{acyclic}} is 
\comp{$\mathsf{\Sigma^p_{k}}$-complete} in the size of the Kripke structure.
 For $k \geq 1$,  \PR{\mbox{(AE)}$k$\mbox{-HyperLTL}}{\mbox{acyclic}} is
\comp{$\mathsf{\Sigma^p_{k+1}}$-complete} in the size of the Kripke structure.

\end{theorem}
 
\noindent {\em  Proof sketch.} \ 
For the upper bound, suppose that the first quantifier is existential.
Since the Kripke structure is acyclic, the length of the traces is bounded by 
the number of states. We can thus nondeterministically guess the repair and 
the existentially quantified traces in polynomial time, and then verify the 
correctness of the guess by model checking the remaining formula, which has 
$k-1$ quantifier alternations and begins with a universal quantifier. The 
verification can be done in \comp{$\mathsf{\Pi^p_{k-1}}$}~\cite[Theorem 
3]{bf18}. Hence, the repair problem is in \comp{$\mathsf{\Sigma^p_{k}}$}. 
Analogously, if the first quantifier is universal, the model checking problem 
in \comp{$\mathsf{\Pi^p_{k}}$} and the repair problem in 
\comp{$\mathsf{\Sigma^p_{k+1}}$}.

We establish the lower bound via a reduction from the {\em 
quantified Boolean formula} (QBF) satisfiability problem~\cite{gj79}.
The Kripke structure (see~Fig.~\ref{fig:system-acyclic-qbf}) contains a 
path for each clause, and a separate structure that consists of a sequence of 
diamond-shaped graphs, one for each variable. A path through the diamonds 
selects a truth value for each variable, by going right or left, respectively, 
at the branching point.

\begin{figure}[t]
%\vspace{-.8cm}
\centering
\scalebox{.8}{
\pgfdeclarelayer{bg}    % declare background layer
\pgfsetlayers{bg,main}  % set the order of the layers (main is the standard

\begin{tikzpicture}

\coordinate (init) at (0, 0);

\node[draw,circle,text width=0.3cm,fill=white] (initstate) at ($ (init) + 
(3, 1) $) {};
\draw[->] ($ (initstate) + (-.5, .5) $) -- (initstate);

\node[draw,circle,text width=0.3cm,fill=black!60] (sh0) at (init) {};

\foreach \i/\j in {0/1,1/2, 2/3}
{

\draw [fill=white] ($ (sh\i) + (-1,-1) $) ellipse (.65cm and .3cm) node (s\i)
{$\{q^{\j}, p\}$};

\draw [fill=white]($ ({sh\i}) + (1,-1) $) ellipse (.65cm and .3cm) node (sb\i) 
{$\{q^{\j}, \bar{p}\}$};

\node[draw,circle,text width=0.25cm,fill=black!60] (sh\j) at ($ (sh\i) +
(0,-2) $)  {};

\draw [->] (sh\i) -- (s\i);
\draw [->] (sh\i) -- (sb\i);
\draw [->] (s\i) -- (sh\j);
\draw [->] (sb\i) -- (sh\j);

}

\path (sh3) edge [loop below] (1);

\foreach \i/\j in {0/1,1/2}
{

\node[draw,circle,text width=0.25cm,fill=black] (u\i0) at ($ (init) + 
(\j*3+.5,0) $) [text=white]{$\hspace*{-0.1cm}\{c\}$};

\foreach \x/\y in {0/1,1/2, 2/3}
{

\ifthenelse{\i = 0}
{
\ifthenelse{\x = 0}
{
\draw [fill=white]($ (u\i\x) + (0,-1) $) ellipse (.65cm and .3cm) node (v\i\x)
{$\{q^{\y}, p\}$};
}{}
\ifthenelse{\x = 1}
{
\draw [fill=white]($ (u\i\x) + (0,-1) $) ellipse (.65cm and .3cm) node (v\i\x)
{$\{q^{\y}, \bar{p}\}$};
}{}
\ifthenelse{\x = 2}
{
\draw [fill=white]($ (u\i\x) + (0,-1) $) ellipse (.65cm and .3cm) node (v\i\x)
{$\{q^{\y}, p\}$};
}{}
}{}

\ifthenelse{\i = 1}
{
\ifthenelse{\x = 0}
{
\draw [fill=white]($ (u\i\x) + (0,-1) $) ellipse (.65cm and .3cm) node (v\i\x)
{$\{q^{\y}, \bar{p}\}$};
}{}
\ifthenelse{\x = 1}
{
\draw [fill=white]($ (u\i\x) + (0,-1) $) ellipse (.65cm and .3cm) node (v\i\x)
{$\{q^{\y}, {p}\}$};
}{}
\ifthenelse{\x = 2}
{
\draw [fill=white]($ (u\i\x) + (0,-1) $) ellipse (.65cm and .3cm) node (v\i\x)
{$\{q^{\y}, \bar{p}\}$};
}{}
}{}

\node[draw,circle,text width=0.25cm,fill=black!60] (u\i\y) at ($ (v\i\x) +
(0,-1) $)  {};

\draw [->] (u\i\x) -- (v\i\x);
\draw [->] (v\i\x) -- (u\i\y);
}

\path (u\i3) edge [loop below] (\i);

}

\draw [->] (initstate) -- (sh0);
\draw [->] (initstate) -- (u00);
\draw [->] (initstate) -- (u10);

\begin{pgfonlayer}{bg}

\draw [rounded corners=10,,dashed,fill=black!15] ($ (init) + (-2,.5) $) 
rectangle ++(4,-7.5) node (r1) {};
\node (text1) at ($ (init) +  (-1,.8) $) {$\pi_d$ traces};

\draw [rounded corners=10,,dashed,fill=black!5] ($ (init) + (2.5,.5) $) 
rectangle ++(5,-7.5) node (r2) {};
\node (text1) at ($ (init) +  (7,.8) $) {$\pi'$ traces};

\end{pgfonlayer}

\end{tikzpicture}
}
\caption{Kripke structure for the formula $y = \exists x_1.\forall 
x_2.\exists x_3.(x_1 \vee \neg x_2 \vee x_3) \wedge 
(\neg x_1 \vee x_2 \vee \neg x_3)$.}
\label{fig:system-acyclic-qbf}
%\vspace{-.5cm}
\end{figure}
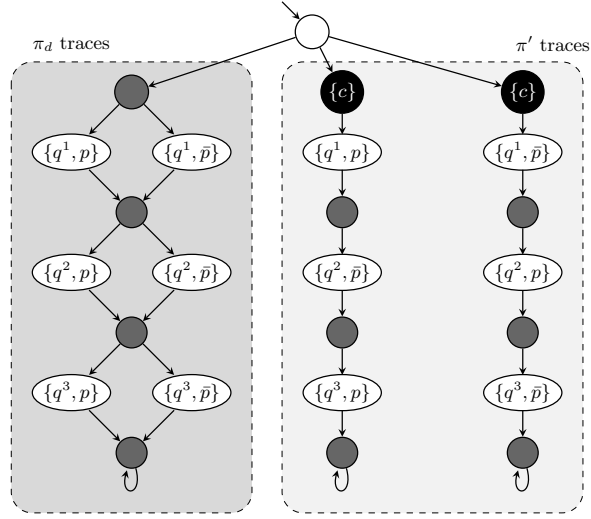

In our reduction, the quantifiers in the QBF instance are translated to trace 
quantifiers (one per alternation depth), resulting in a HyperLTL formula with 
$k$ quantifier alternations and a leading 
existential quantifier. Note that, the outermost existential quantifiers are 
not translated to a quantifier, but instead 
resolved by the repair. For this reason, it suffices to build a HyperLTL 
formula with one less quantifier alternation than the original QBF instance.
Also, in our mapping, we must make sure that the clauses and the diamonds for 
all variables except the outermost existential variables are not removed during 
the repair. Similar to the proof of Theorem~\ref{thrm:sys-tree-aae-lower}, we 
add a counter to the clauses and add a constraint to the HyperLTL formula that 
ensures that all counter values are still present in the repair; for the 
diamonds of the variables, the valuations themselves form such a counter, and 
we add a constraint that ensures that all valuations for the variables (except 
for the outermost existential variables) are still present in the 
repair.\qed

\medskip

Finally, Theorem~\ref{thrm:system-acyc-EAk1} implies that the repair 
problem
for acyclic Kripke structures and HyperLTL formulas with
an arbitrary 
number of quantifiers is in \comp{PSPACE}.

\begin{corollary}
\label{cor:sys-acyclic-hltl}
\PR{{HyperLTL}}{\mbox{acyclic}} is in \comp{PSPACE} in the size of the 
Kripke structure.
 
\end{corollary}

\section{Complexity of Repair for General Graphs}
\label{sec:general}

In this section, we investigate the complexity of the repair problem for 
general graphs. We again begin with the alternation-free fragment and then continue with formulas with quantifier alternation.

\subsection{The Alternation-free Fragment}

We start with the existential fragment. Similar to the case of acyclic graphs, 
the repair problem can be solved with a model checking algorithm.

\begin{theorem}
\label{thm:sys-general-e}
\PR{$\mbox{E}^*$-HyperLTL}{\mbox{general}} is \comp{NL-complete} in the size of 
the 
Kripke structure.
\end{theorem}

\begin{proof}
  Analogously to the proof of Theorem~\ref{thm:sys-acyc-e}, we note that, for 
existential formulas, the repair problem is equivalent to the model checking 
problem. A given Kripke structure satisfies the formula if and only if it has a 
repair. If the formula is satisfied, the repair is simply the original Kripke 
structure. Since the model checking problem for existential formulas for 
general graphs is \comp{NL}-complete~\cite{frs15}, the same holds for the 
repair problem.\qed
\end{proof}

Unlike the case of acyclic graphs, the repair problem for the universal 
fragment is more expensive, although the model checking problem is 
\comp{NL-complete}~\cite{bf18}. 

\begin{theorem}
\PR{$\mbox{A}^+$-HyperLTL}{\mbox{general}} is \comp{NP-complete} in the size 
of the Kripke structure.
\label{thm:sys-general-a}
\end{theorem}

\begin{proof} For membership in \comp{NP}, we nondeterministically guess a 
solution to 
the repair problem, and verify the correctness of the universally 
quantified HyperLTL formula against the solution in polynomial time in the size 
of the Kripke structure. 
\comp{NP}-hardness follows from the \comp{NP-hardness} of the repair 
problem for LTL~\cite{bek09}.
%in particular a formula of the form:
%
%$$\forall \pi. \Big(\G\F p_{\pi} \, \wedge \, \G\F q_{\pi}\Big).$$
%
%This means that problem is \comp{NP-complete}.

\end{proof}

\subsection{Formulas with Quantifier Alternation}

Next, we consider formulas where the number of quantifier
alternations is 
bounded by a constant $k$. We show that changing
the frame structure from 
acyclic to general graphs again results in a
significant increase in complexity (see 
Table~\ref{tab:system}).

\begin{theorem}
  \label{thrm:system-general-EAk}
\PR{$\mbox{E}^*\mbox{A}^*\mbox{-HyperLTL}$}{\mbox{general}} is in
\comp{PSPACE} in the size of the Kripke structure.
\PR{$\mbox{A}^*\mbox{E}^*\mbox{-HyperLTL}$}{\mbox{general}} is
\comp{PSPACE}-complete in the size of the Kripke structure.
  For $k \geq 2$,
\PR{$\mbox{(EA)}^k\mbox{-HyperLTL}$}{\mbox{general}} and
\PR{$\mbox{(AE)}^k\mbox{-HyperLTL}$}{\mbox{general}}
are
\comp{$(k{-}1)$-EXPSPACE}-complete in the size of the Kripke structure.
\end{theorem}

\noindent {\em Proof idea.} \
The claimed complexities are those of the model checking 
problem~\cite{markus}.
  We prove that the repair problem has the same complexity as the model 
checking problem.
  To show the upper bound of 
\PR{$\mbox{A}^*\mbox{E}^*\mbox{-HyperLTL}$}{\mbox{general}}, we enumerate, in 
\comp{PSPACE}, all possible repairs, 
and then verify against the HyperLTL formula.

  For the lower bounds, we modify the Kripke structure and the HyperLTL formula 
such that the only possible repair is the unchanged Kripke structure.
  After the modification, the repair problem thus has the same result as the 
model checking problem.
  The idea of the modification is to assign numbers to the successors of each 
state. We add extra states such that the
  traces that originate from these states correspond to all possible number 
sequences. Finally, the HyperLTL formula states
  that for each such number sequence there exists a corresponding trace in the 
original Kripke structure.
\qed

\medskip

Finally, Theorem~\ref{thrm:system-general-EAk} implies that the repair 
problem
for general Kripke structures and HyperLTL formulas with
an arbitrary 
number of quantifiers is in \comp{NONELEMENTARY}.

\begin{corollary}
\label{cor:sys-general-hltl}
\PR{{HyperLTL}}{\mbox{general}} is \comp{NONELEMENTARY} in the size of the 
Kripke structure. 
\end{corollary}

\section{Related Work}
\label{sec:related}

%The notion of hyperproperties was introduced by Clarkson and 
%Schneider~\cite{cs10}. HyperLTL~\cite{cfkmrs14} is a temporal logic for 
%hyperproperties.

There has been a lot of recent progress in automatically
\emph{verifying}~\cite{frs15,FinkbeinerMSZ-CCS17,10.1007/978-3-319-96145-3_8,10.1007/978-3-030-25540-4_7}
and \emph{monitoring}~\cite{ab16,Finkbeiner2019,bsb17,bss18,fhst18,sssb19,10.1007/978-3-030-17465-1_7} HyperLTL 
specifications.
HyperLTL is also supported by a growing 
set of tools, including the model checker MCHyper~\cite{frs15,10.1007/978-3-030-25540-4_7}, the  satisfiability checkers 
EAHyper~\cite{fhs17} and MGHyper~\cite{10.1007/978-3-030-01090-4_31}, and the runtime monitoring tool 
RVHyper~\cite{fhst18}.

Directly related to the repair problem studied in this paper are the 
satisfiability and synthesis problems. The \emph{satisfiability} problem for 
HyperLTL was shown to be 
decidable for the $\exists^*\forall^*$ 
fragment and for any fragment that includes a $\forall\exists$ quantifier alternation~\cite{fh16}. The hierarchy of hyperlogics beyond HyperLTL has been studied in~\cite{hierarchy}.

The \emph{synthesis} problem was shown to be undecidable in general, and 
decidable for the $\exists^*$ and $\exists^*\forall$ fragments. While the 
synthesis problem becomes, in general, undecidable as soon as there are two 
universal quantifiers, there is a special class of universal specifications, 
called the linear $\forall^*$-fragment, which is still decidable~\cite{fhlst18}. 
The linear $\forall^*$-fragment corresponds to the decidable \emph{distributed 
synthesis} problems~\cite{conf/lics/FinkbeinerS05}. 
The \emph{bounded synthesis} problem considers only systems up to a given bound on the number of states.
Bounded synthesis from hyperproperties is studied 
in~\cite{fhlst18,10.1007/978-3-030-25540-4_7}. Bounded synthesis has been successfully applied to small examples such as the dining cryptographers~\cite{journals/cacm/Chaum85}.

The problem of model checking hyperproperties for tree-shaped and acyclic graphs 
was studied in~\cite{bf18}. Earlier, a similar study of the impact of 
structural restrictions on the complexity of the model checking problem has 
also been carried out for LTL~\cite{kb11}.

For LTL, the complexity of the repair 
problem was studied independently in~\cite{bek09,ekb05,jgb05} and 
subsequently in~\cite{bk08} for distributed programs. The repair problem is also related to \emph{supervisory control}, where, for a given plant, a supervisor is constructed that selects an appropriate subset of the plant's controllable actions to ensure that the resulting behavior is safe~\cite{doi:10.1080/00207178608933645,187347,doi:10.1137/S0363012902409982}. 

%The LTL model checking problem is 
%\comp{PSPACE-hard} if there exists a strongly connected component with two 
%distinct cycles in the Kripke structure. If no such component exists, then the 
%model checking problem is in \comp{coNP}. For the special case
%of finite paths and trees, the LTL model checking problem is
%in \comp{NC}, or, more precisely, in \comp{AC$^1$(logDCFL)}~\cite{kf09,lars}.

% As mentioned in the introduction, refinement techniques for information-flow 
% properties such as noninterference have been proposed in the 
% literature~\cite{j89,m01,gs91,o90,rww96,ss06}. These techniques prescribe 
% stepwise manual or semi-automated deductive refinement methods. 
% More recently, type-theoretic refinement approaches~\cite{bgjahrs15} have also 
% been introduced.
% \todo{Add Rajeev's paper.}

%Runtime verification of HyperLTL was first studied in~\cite{ab16}. The authors
%define the notion of monitorability for HyperLTL and propose an algorithm for 
%monitoring alternation-free dijunctive HyperLTL formulas. This approach was 
%then generalized to the full alternation-free fragment of HyperLTL 
%in~\cite{bsb17} using a rewriting-based technique. More recently, 
%in~\cite{fhst17}, the authors propose an automata-based monitoring algorithm 
%and identify class of formulas for which monitoring can be done in a 
%space-efficient manner.

\section{Conclusion and Future Work}
\label{sec:conclusion}

In this paper, we have developed a detailed classification of 
the complexity of the repair problem for hyperproperties 
expressed in HyperLTL. We considered general, acyclic, and tree-shaped Kripke 
structures. We showed that for trees, the complexity of the repair 
problem in the size of the Kripke structure does not go beyond \comp{NP}. The 
problem is complete for \comp{L}, \comp{P}, and \comp{NP} for 
fragments with only one quantifier alternation, depending upon the outermost 
quantifiers. For 
acyclic Kripke structures, the complexity is in \comp{PSPACE} (in the level of 
the polynomial hierarchy that corresponds to the number of quantifier 
alternations). The problem is \comp{NL-complete} for the 
alternation-free fragment. For general graphs, the problem is 
\comp{NONELEMENTARY} for an arbitrary number of quantifier alternations. 
For a bounded number $k$ of alternations, the problem is 
$(k{-}1)$-\comp{EXPSPACE-complete}. 
These results highlight a crucial insight to the repair problem compared 
to the corresponding model checking problem~\cite{bf18}. With the notable exception of trees, 
where the complexity of repair is \comp{NP-complete}, compared to 
the \comp{L-completeness} of model checking, the complexities of repair and model checking are largely aligned. 
This is mainly due to the fact that computing a repair can be done 
by identifying a candidate substructure, which is comparatively 
inexpensive, and then verifying its correctness. 
% This also means that the 
% substantial differences between the complexities reported in 
% Table~\ref{tab:system}, in particular the contrast to the non-elementary 
% complexity of the most general problem, are intriguing. These results suggest 
% that for a large class of systems (tree and acyclic graphs) and also 
% non-exhaustive techniques such as example-based synthesis~\cite{at17} that work 
% on restricted structures, we have a significant complexity advantage over 
% brute-force synthesis. Also, one can view the (bounded) synthesis problem as a refinement 
% problem that starts from a complete (general, acyclic, or tree) graph.

The work in this paper opens many new avenues for further research. An
immediate question left unanswered in this paper is the lower bound
complexity for the $\exists^*\forall^*$ fragment in acyclic and general
graphs. 
% Also, there are extensions of HyperLTL, such as the
% branching-time logic HyperCTL$^*$~\cite{cfkmrs14} and the first-order
% extension FOHLTL~\cite{FinkbeinerMSZ-CCS17} as well as the
% probabilistic variation HyperPCTL~\cite{ab18}, for which one can study
% the refinement problem.
It would be interesting to see if the
differences we observed for HyperLTL carry over to other hyperlogics (cf. 
\cite{hierarchy,cfkmrs14,FinkbeinerMSZ-CCS17,ab18}). One could extend the 
results of this paper to the reactive setting, where the program interacts 
with the environment. And, finally, the ideas of this paper 
might help to extend popular synthesis techniques for general (infinite-state) 
programs, such as program sketching~\cite{s13} and syntax-guided 
synthesis~\cite{abjmrssstu13}, to 
hyperproperties.

\section*{Acknowledgments}

We would like to thank Sandeep Kulkrani, Reza Hajisheykhi (Michigan State 
University), and Shreya Agrawal (Google). The repair problem was 
originally inspired through our interaction with them. 
This work is sponsored in part by NSF SaTC Award 1813388. It was 
also 
supported by the German Research Foundation (DFG) as part of the Collaborative 
Research Center “Methods and Tools for Understanding and Controlling Privacy” 
(CRC 1223) and the Collaborative Research Center “Foundations of Perspicuous 
Software Systems” (TRR 248, 389792660), and by the European Research Council 
(ERC) Grant OSARES (No. 683300).

\bibliographystyle{plain}
\bibliography{bibliography}

\end{document}